\documentclass[letterpaper]{article}
\usepackage{aaai17}
\usepackage{hyperref}

\usepackage[utf8x]{inputenc}
\usepackage{times}
\usepackage{helvet}
\usepackage{courier}
\usepackage{amssymb}
\usepackage{amsmath}
\usepackage{graphicx}
\usepackage{tikz}
\usepackage{numprint}
\usepackage{lipsum}
\usepackage{mathtools}

\usepackage{theorem}

\usepackage{amsthm}

\usepackage{thmtools}
\usepackage{thm-restate}
\usepackage{amsopn}

\usepackage{algorithm}
\usepackage{algpseudocode}

\newcommand*\circled[1]{%
  \tikz[baseline=(C.base)]\node[draw,circle,inner sep=0.5pt](C) {#1};\!
}

\newtheorem{thm}{Theorem}

\newtheorem{lem}{Lemma}
\newtheorem{rem}{Remark}

\newtheorem{defn}{Definition}
\newtheorem{exam}{Example}

\algtext*{EndWhile}
\algtext*{EndIf}
\algtext*{EndFor}
\algtext*{EndProcedure}

\newcommand{\eps}{\epsilon}
\newcommand{\bigoh}{O}
\DeclareMathOperator{\cof}{cof}
\DeclareMathOperator{\coeff}{coeff}

\title{Extensive-Form Perfect Equilibrium Computation in Two-Player Games}
\author{Gabriele Farina\\
  Computer Science Department\\
  Carnegie Mellon University\\
  5000 Forbes Avenue\\
  Pittsburgh, PA 15213, USA\\
  \href{mailto:gfarina@cs.cmu.edu}{\texttt{gfarina@cs.cmu.edu}}
\And Nicola Gatti\\
Dipartimento di Elettronica, Informazione e Bioingegneria\\Politecnico di Milano\\
Piazza Leonardo da Vinci, 32\\
I-20133, Milan, Italy\\
  \href{mailto:nicola.gatti@polimi.it}{\texttt{nicola.gatti@polimi.it}}}
\newif\ifshowsupplementalmaterial
\showsupplementalmaterialfalse

\begin{document}
	\maketitle
\begin{abstract}
\begin{quote}
We study the problem of computing an Extensive-Form Perfect Equilibrium (EFPE) in 2-player games. This equilibrium concept refines the Nash equilibrium requiring resilience w.r.t. a specific vanishing perturbation (representing  mistakes of the players at each decision node). The scientific challenge is intrinsic to the EFPE definition: it requires a perturbation over the agent form, but the agent form is computationally inefficient, due to the presence of highly nonlinear constraints. We show that the sequence form can be exploited in a non-trivial way and that, for general-sum games, finding an EFPE is equivalent to solving a suitably perturbed linear complementarity problem. We prove that Lemke's algorithm can be applied, showing that computing an EFPE is \textsf{PPAD}-complete. In the notable case of zero-sum games, the problem is in \textsf{FP} and can be solved by linear programming. 
Our algorithms also allow one to find a Nash equilibrium when players cannot perfectly control their moves, being subject to a given execution uncertainty, as is the case in most realistic physical settings.
\end{quote}
\end{abstract}

\section{Introduction}

Computing solutions of games is currently one of the hottest problems in computer science, as providing optimal strategies to \emph{autonomous agents} interacting strategically is central in Artificial Intelligence~\cite{Shoham:2008:MSA:1483085}. Finding a Nash Equilibrium (NE)---the basic solution concept for non-cooperative games---is \textsf{PPAD}-complete even in 2-player games~\cite{Chen:2009:SCC:1516512.1516516} and it is unlikely that there is a polynomial-time algorithm, since it is commonly believed that $\mathsf{FP}\subset\mathsf{PPAD}\subset\mathsf{FNP}$. We recall a search problem is in the \textsf{PPAD} class if there is a \emph{path-following} algorithm whose iterations have a polynomial-time cost. In the case of 2-player normal-form games, this algorithm is provided by~\cite{lemkehowson}. 

Extensive-form games provide a richer representation of strategic interaction situations w.r.t. the normal form. The study of extensive-form games is much more involved than that of normal-form games. A variation of Lemke-Howson's algorithm, called \emph{Lemke's algorithm}, finds an NE in a 2-player extensive-form game showing that the problem is in the \textsf{PPAD} class~\cite{koller1996efficient}. However, the concept of NE is not satisfactory in extensive-form games, and NE refinements are studied~\cite{selten1975efpe}. When information is perfect, the concept of Subgame Perfect Equilibrium (SPE) is satisfactory, while it is not when information is imperfect. In this latter case, refinements  are usually based on the idea of \emph{perturbations} representing mistakes of the players. In a Quasi-Perfect Equilibrium (QPE)---proposed by van Damme--- a player maximizes their utility in each decision node taking into account \emph{only} the future mistakes of the opponents, whereas, in an Extensive-Form Perfect Equilibrium (EFPE)---proposed by Nobel prized Selten---, players maximize their utility in each decision node keeping into account the future mistakes of \emph{both} themselves and their opponents~\cite{Hillas20021597}. The sets of QPEs and EFPEs may be disjoint, requiring different techniques. Given a specific perturbation, computing a QPE is \textsf{PPAD}-complete~\cite{quasiperfect} and can be done by summing the perturbation to the constant terms in the linear constraints of the sequence form~\cite{VONSTENGEL1996220}; due to this reason, we say that this pertubation is \emph{additive}. However, the problem of efficiently computing an EFPE is still open. The scientific challenge is intrinsic to the EFPE definition: it is based on a perturbation over the agent form, but the agent form is computationally inefficient, presenting highly non-linear equilibrium constraints. The only previous attempt is~\cite{DBLP:conf/aaai/GattiI11}, but no proof is provided about neither the soundness  nor polynomial-time cost of each algorithm iteration (details are in the Supplemental Material).

	We show that finding an EFPE is \textsf{PPAD}-complete in 2-player general-sum games and can be done by means of Lemke's algorithm with an extra polynomial computation cost due to a numeric perturbation, and that it is in \textsf{FP} in 2-player zero-sum games and can be done by linear programming with the same perturbation for the general-sum case. The table below summarizes the results known so far. `($*$)` denotes original contribution discussed in this paper.

				\begin{table}[H]
\resizebox{\linewidth}{!}{
		  \begin{tabular}{|l|l|l|} \hline
		  	Solution concept                       &  General-sum &  Zero-sum\\\hline \hline
		  	Nash (NE)             &  \textsf{PPAD}-complete        &  \textsf{FP}\\\hline
		  	Subgame Perfect (SPE) &  \textsf{PPAD}-complete        &  \textsf{FP}\\\hline
		  	Quasi Perfect (QPE)   &  \textsf{PPAD}-complete        &  \textsf{FP}\\\hline
		  	Extensive-Form Perfect (EFPE) &  \textsf{PPAD}-complete ($*$) &  \textsf{FP} ($*$)\\\hline
		  \end{tabular}
}
		\end{table}
	
In order to prove our main result, we provide  also two  original results of broader interest. First, we show that a perturbation over the agent form can be formulated as a specific symbolic perturbation over the coefficients of the variables of the sequence form (due to this reason, we say that this perturbation is \emph{multiplicative}). This shows that computing an equilibrium when a player does not have perfect control over the execution of their moves along the game tree, as is customary for physical agents (e.g., robots) whose actions are subject to execution uncertainty, is \textsf{PPAD}-complete or in \textsf{FP} in general-sum and zero-sum games, respectively. Second, we show that we can turn the symbolically perturbed problem above into a numerically perturbed problem. We believe our approach to be particularly interesting, in that it not only applies to the computation of EFPEs, but rather is a more general framework, that can be used to derive, e.g., the results on QPEs in a more natural fashion. All omitted proofs can be found in the Supplemental Material.

\section{Preliminaries}
	In the following, we adopt the notation introduced by~\cite{Shoham:2008:MSA:1483085}. We invite the reader unfamiliar with the topic to refer to~\cite{Shoham:2008:MSA:1483085} or any other classic textbook on the subject for further information and context.
	
	An \emph{extensive-form game} $\Gamma$ is defined over a game tree. In each non-terminal node a single player moves and each edge corresponds to an action available to the player. As customary, $N$ denotes the set of players, $A_i$ denotes the set of actions available to player~$i$ and $a$ is an action, $\mathbf{a}$ denotes the action profile of all the players  and $\mathbf{a}_{-i}$ denotes the action profile of the opponents of player~$i$. Furthermore, $H_i$ denotes the set of information sets of player~$i$ and $h$ is an information set. Finally, $\iota(h)$ is the player that moves at~$h$, $\rho(h)$ is the set of actions available at~$h$ to player~$\iota(h)$, and function $u_i$ returns the utility of player~$i$ from each terminal node.

	The agent form~\cite{selten1975efpe} of an extensive-form game is a tabular representation in which, for every player~$i$ and information set~$h \in H_i$, there is a fictitious player called \emph{agent} and all the agents of player~$i$ have the same utility from the terminal nodes. Player~$i$'s strategy over action~$a$, called \emph{behavioral}, is denoted by $\pi_i(a)\geq 0$ and is such that for each $h$ it holds $\sum_{a \in \rho(h)}\pi_{\iota(h)}(a)=1$. The strategy of the agent playing at $h$ is the restriction of $\pi_{\iota(a)}$ to actions $\rho(h)$. A behavioral strategy profile is denoted by~$\pi$.

	The concept of Extensive-Form Perfect Equilibrium~\cite{selten1975efpe}, also known as ``Trembling hand perfect equilibrium'', is defined on the agent form. We initially introduce the definitions of \emph{perturbed game} (over the agent form) and Nash equilibrium of the agent form since they are necessary to introduce the definition of EFPE.
	\begin{defn}\label{def:perturbedgame}
		Let $\Gamma$ be an extensive-form game and $l(a)>0$ be a positive number called \emph{perturbation} such that $\sum_{a \in \rho(h)}l(a)<1$ for every $h$, then a (agent-form) perturbed game $(\Gamma, l)$ is an extensive-form game with the constraint that  $\pi_{\iota(h)}(a)\ge l(a)$ for every $h$ and $a \in \rho(h)$.
	\end{defn}
	\begin{defn}\label{def:efne}
		A behavioral strategy profile $\pi$ is a Nash equilibrium of the agent form of $\Gamma$ if, for every information set~$h$, the behavioral strategy of the agent playing at~$h$ is best response to the strategies of all the other agents.
	\end{defn}
	The problem of finding a Nash equilibrium of the agent form can be formulated as a non-linear complementarity problem (NLCP). This formulation is not useful in practice since the high non-linearity raises a number of computational issues. We now introduce the definition of EFPE.
	\begin{defn}\label{def:efpe}
		A strategy profile $\pi$ is an EFPE of $\Gamma$ if it is a limit point of a sequence $\{\pi(l)\}_{l \downarrow 0}$ where $\pi(l)$ is a Nash equilibrium of the agent form of the perturbed game $(\Gamma, l)$.
	\end{defn}
	
	Finally, we introduce the sequence form~\cite{VONSTENGEL1996220} that provides a computationally efficient representation of an extensive-form game. The set of players of the sequence form is the same of that of the extensive form and each player~$i$ plays \emph{sequences} $q\in Q_i$ of actions $a \in A_i$ over the game tree. There is a special sequence, denoted by~$q_\emptyset$ and available to all the players, and all the other sequences $q \in Q_i$ are defined by induction extending some sequence~$q'\in Q_i$, starting from $q_\emptyset$, with an action~$a\in A_i$. As customary, $qa \in Q_i$ denotes the sequence obtained by extending sequence~$q\in Q_i$ with  action~$a\in A_i$. A sequence is called \emph{terminal} if, combined with some sequence of the other players, leads to a terminal node, and \emph{non-terminal} otherwise. With 2 players, $U_i$ is the utility matrix of player~$i$ and $U_i(q_i,q_{-i})$ returns, when sequence profile $(q_i,q_{-i})$ leads to a terminal node (here $q_i \in Q_i$ and $q_{-i} \in Q_{-i}$), the utility of the node, and zero otherwise. The strategy of player~$i$ over sequence $q$ is denoted by $r_i(q)\ge 0$ and is called \emph{realization plan}. Finally, strategies are subject to special constraints: $r_i(q_\emptyset)=1$ and, for every information set~$h$ and sequence~$q$ leading to~$h$, $r_i(q) = \sum_{a \in \rho(h)}r_i(qa)$. For notational convenience, these constraints can be written as: $F_i\, r_i = f_i$, where $F_i$ is an opportune matrix and $f_i$ is a vector of zeros except for the first position whose value is one. Finally, we recall that a strategy profile $r$ is \emph{realization equivalent} to a strategy profile $\pi$ when $r$ and $\pi$ induce the same probability distribution on the terminal nodes. Given profile~$r$, a realization-equivalent~$\pi$ can be derived as $\pi_i(a) = r_i(qa)/r_i(q)$ if $r_i(q)>0$ and $\pi_i(a)$ is any otherwise.

	The definition of a Nash equilibrium of the sequence form is standard, requiring each player to play their best response. In contrast to what happens in the agent form, the problem of finding a Nash equilibrium in the sequence form can be formulated as a linear complementarity problem (LCP)~\cite{koller1996efficient} and can be solved by means of Lemke's algorithm~\cite{lemke1970recent}. In particular, \cite{koller1996efficient} show that by applying an opportune affine transformation of the players' utility matrices the LCP satisfies two properties that allow  Lemke's algorithm to terminate always with a Nash equilibrium (we report these two properties, that we use in the following, in the Supplemental Material). This result, combined with the fact the computational cost of each pivoting step of Lemke's algorithm is polynomial, shows that the problem of finding a Nash equilibrium of the sequence form is in the \textsf{PPAD} class. Importantly, it is known that, without any perturbation, any Nash equilibrium of the sequence form is also a Nash equilibrium of the agent form, while, in presence of perturbations, this result may hold or not, depending on the definition of the specific perturbation used.

\section{Extensive-Form Perfect Equilibria and LCPs}
We initially show that introducing a specific perturbation over the realization-plan strategies is equivalent to introducing a perturbation over the behavioral strategies. For the sake of simplicity, we study the case in which $l(a)=\eps$ for every $a$, thus leading to a specific EFPE. All the results discussed in this section and in the following ones can be extended to the general case in which $l(a)$ is a polynomial in $\eps$ potentially different for each action~$a$
---we recall that considering only polynomial functions of $\eps$ is sufficient to find any EFPE, as discussed in~\cite{blum1991,govindam2003}.
	\begin{thm}\label{thm:efpesequence}
		A realization-plan strategy profile $r$ is an EFPE of $\Gamma$ if it is a limit point of a sequence $\{r(\eps)\}_{\eps \downarrow 0}$, where $r(\eps)$ is a Nash equilibrium of the sequence form of $\Gamma$ under the constraint $r_i(qa) \ge \eps r_i(q)$ for every player~$i$, sequence~$q$, and action~$a$.
	\end{thm}
	\begin{proof}
		The proof is structured into two steps. In the first step, we show that requiring $r_i(qa) \ge \eps r_i(q)$ for every player~$i$, sequence~$q$, and action~$a$ in sequence form is equivalent to considering the perturbed game $(\Gamma, l)$ where $l(a)=\eps$ for every action~$a$. 
In the second step, we show that any Nash equilibrium in the sequence form of such a perturbed game is a Nash equilibrium in the agent form.

		Focus on the first step. The empty sequence $q_\varnothing$ is played with probability one and then, by induction, every sequence $q$ is played with a strictly positive probability of at least $\eps^{|q|}$ where $|q|$ is the length in terms of actions of sequence $q$. Since the behavioral strategy $\pi_i(a)$ is defined as $\pi_i(a) = r_i(qa)/r_i(q)$, we have that requiring $r_i(qa) \ge \eps r_i(q)$ is equivalent to require $\pi_i(a) \ge \eps$ for every $a$. This completes the proof of the first step.

		Focus on the second step. The proof follows from the definition of sequence form. Nevertheless, we report all the details. The expected utility (in the agent-form representation) $EU^{\text{AF}}_{a_h}$ provided by action $a_h \in \rho(h)$ to player $\iota(h)=i$ given $\pi_{-h}$ is:
		\begin{multline*}
		EU^{\text{AF}}_{a_h}(\pi_{-h}) = \\ \sum\limits_{\mathbf{a}_{-h} \in A_{-h}} U^{\text{AF}}_i(a_h, \mathbf{a}_{-h}) \prod\limits_{h' \in H\setminus\{h\}}\pi_{\iota(h')}((\mathbf{a}_{-h})_{h'}).
		\end{multline*}
		We denote by $U^{\text{AF}}_i$ the utility function of player~$i$ in the agent-form representation and by $\mathbf{a}_{-h}$ the action profile in which only the action played at~$h$ is excluded.
		The expected utility (in the sequence-form representation) $EU^{\text{SF}}_{a_h}$ provided by sequence $qa_h \in Q_i$ where $a_h \in \rho(h)$ to player $i$ given $r_{-i}$ is:
		\[
		EU^{\text{SF}}_{a_h}(r_{-i}) = \sum\limits_{q': qa_h \in q'}\sum\limits_{q'' \in Q_{-i}} U_i(q',q'')r_{-i}(q').
		\]
		In the agent form, for each information set~$h$, a Nash equilibrium assures that action~$a_h$ is played with $\sigma_{\iota(h)}(a_h)> \eps$
		only if $EU^{\text{AF}}_{a_h}(\pi_{-h})$ is the maximum among the $EU^{\text{AF}}_{a'_h}(\pi_{-h})$s for all $a'_h\in \rho(h)$. In the sequence form, for each information set~$h$, a Nash equilibrium assures that sequence~$qa_h$ is played with $r_{\iota(h)}(qa_h)> \eps r_{\iota(h)}(a_h)$
		only if $EU^{\text{SF}}_{qa_h}(r_{-i})$ is the maximum among the $EU^{\text{SF}}_{qa'_h}(r_{-i})$s for all $a'_h\in \rho(h)$. We show that the two families of constraints are the same except for an affine transformation not depending on the actions available at the information set~$h$ and preserving the maximum. At every information set~$h$, it holds $EU^{\text{AF}}_{a_h}(\pi_{-h}) = \alpha EU^{\text{SF}}_{a_h}(r_{-i}) + \beta$ for every $a_h \in \rho(h)$, where $\alpha_h$ and $\beta_h$ do not depend on the actions available at~$h$. More precisely, $\alpha = \prod_{a' \in q}\pi_i(a') = r_i(q)> \eps^{|q|}>0$ and $\beta = \sum\limits_{q': \not \exists a' \in \rho(h), qa' \in q'}\sum\limits_{q'' \in Q_{-i}} U_i(q',q'')r_{-i}(q')$. Therefore, the action $a_h$ that maximizes $EU^{\text{AF}}_{a_h}(\pi_{-h})$ maximizes also $EU^{\text{SF}}_{a_h}(r_{-i})$ when $\pi$ and $r$ are realization equivalent.
	\end{proof}

	The proof of the theorem above shows that requiring the condition $r_i(qa) \ge \eps\, r_i(q)$ for every $i \in N,q \in Q_i,a \in A_i$ is equivalent to considering the perturbed game $(\Gamma, l)$ where $l(a) =\eps$ for every action~$a$. For notational convenience, such a condition can be expressed as:
	\[
		R_i(\eps)\, r_i = \tilde{r}_i \geq 0,
	\]
	where $R_i(\eps)$ is a matrix that we call \emph{behavioral perturbation matrix} and $\tilde{r}_i$ is the residual strategy (i.e., the strategy of the player once perturbation has been excluded).

	\begin{figure}
		\centering\begin{tikzpicture}[scale=0.85]
			\node[circle, inner sep=.5mm, fill=black] (p11) at (0, 0) {};
			\node[circle, inner sep=.5mm, fill=black] (l1) at (-1, -1) {};
			\node[circle, inner sep=.5mm, fill=black] (p12) at (1, -1) {};
			\node[circle, inner sep=.5mm, fill=black] (l2) at (0, -2) {};
			\node[circle, inner sep=.5mm, fill=black] (p21) at (2, -2) {};
			\node[circle, inner sep=.5mm, fill=black] (l3) at (1, -3) {};
			\node[circle, inner sep=.5mm, fill=black] (l4) at (3, -3) {};
			\draw (p11) node[above]{\footnotesize $1.1$} --node[above left] {\footnotesize $\mathsf{L}_1$} (l1) node[below] {\footnotesize $(1,1)$};
			\draw (p11) --node[above right] {\footnotesize $\mathsf{R}_1$} (p12) node[above right] {\footnotesize $1.2$};
			\draw (p12) --node[above left] {\footnotesize $\mathsf{L}_2$} (l2) node[below] {\footnotesize $(1,1)$};
			\draw (p12) --node[above right] {\footnotesize $\mathsf{R}_2$} (p21) node[above right] {\footnotesize $2.1$};
			\draw (p21) --node[above left] {\footnotesize $\mathsf{l}_1$} (l3) node[below] {\footnotesize $(1,1)$};
			\draw (p21) --node[above right] {\footnotesize $\mathsf{r}_1$} (l4) node[below] {\footnotesize $(0,0)$};
		\end{tikzpicture}

		\caption{A sample game.}

				\label{fig:sample game}
	\end{figure}
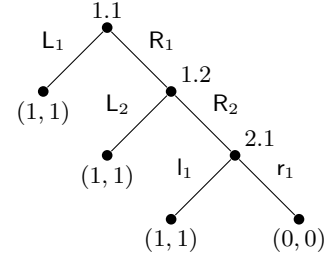

	\begin{exam}\label{exmp:behavioral perturbation matrices}
		Consider the sample game of Figure~\ref{fig:sample game}. Matrices $R_1(\eps)$ and $R_2(\eps)$ are as follows:
		\[
			R_1(\eps) = \begin{pmatrix}
				1 & 0 & 0 & 0 & 0\\
				-\eps & 1 & 0 & 0 & 0\\
				-\eps & 0 & 1 & 0 & 0\\
				0 & 0 & -\eps & 1 & 0\\
				0 & 0 & -\eps & 0 & 1
			\end{pmatrix}\!,\,
			R_2(\eps) = \begin{pmatrix}
				1 & 0 & 0\\
				-\eps & 1 & 0\\
				-\eps & 0 & 1
			\end{pmatrix}.
		\]
	\end{exam}
	We study the properties of matrix $R(\eps)$ and of its inverse.
	\begin{rem}
		Behavioral perturbation matrices are lower triangular square matrices having only $0$ or $-\eps$ as entries.
	\end{rem}
	\begin{restatable}{lem}{inverseR}\label{lem:inverse R}
		Let $R(\eps)$ be a $n\times n$ behavioral perturbation matrix. Then $R(\eps)$ is invertible, and its inverse is
		\[
			R(\eps)^{-1} = I + \eps E(\eps),
		\]
		where $I$ is the identity matrix, and $E(\eps)$ is a lower triangular matrix whose entries are polynomials in $\eps$ having non-negative integer coefficients.
	\end{restatable}

	\begin{exam}
		For the matrix $R_1(\eps)$ of Example~\ref{exmp:behavioral perturbation matrices}, we have:
		\[
			R_1(\eps)^{-1} = \begin{pmatrix*}[l]
				1 & 0 & 0 & 0 & 0\\
				\eps & 1 & 0 & 0 & 0\\
				\eps & 0 & 1 & 0 & 0\\
				\eps^2 & 0 & \eps & 1 & 0\\
				\eps^2 & 0 & \eps & 0 & 1
			\end{pmatrix*}.
		\]
	\end{exam}
	Now we are in the position to formulate the problem of finding an EFPE as a linear complementarity program.
	\begin{lem}\label{lem:EFPE LCP}
		An EFPE is the limit point as $\eps \to 0$ of any solution of the perturbed standard-form LCP
\[			P(\eps)\ :\ \left\lbrace
			\begin{array}{ll}
			  \text{\normalfont{find}} & z, w\\
			  \text{\normalfont s.t.} & z^\top w = 0\\
			  \hfill & w = M(\eps)z + b\\
			  \hfill & z, w \ge 0
			\end{array}\right.
\]
where (the underlined entries depend on $\eps$)
		\[
			z = \begin{pmatrix*}[l]
				\tilde{r}_1\\
				\tilde{r}_2\\
				v_1^+\\
				v_1^-\\
				v_2^+\\
				v_2^-
			\end{pmatrix*},\quad
			b = \begin{pmatrix*}[l]
				\textcolor{white}{-}0\\
				\textcolor{white}{-}0\\
				\textcolor{white}{-}f_1\\
				-f_1\\
				\textcolor{white}{-}f_2\\
				-f_2
			\end{pmatrix*},
		\]\[
			\underline{M} = \small\left(\begin{array}{lll}
				0 & -\underline{R_1}^{-\top}U_1\underline{R_2}^{-1} & \underline{R_1}^{-\top}F_1^\top\\
				-\underline{R_2}^{-\top}U_2^{\top}\underline{R_1}^{-1} & 0 & 0\\
				-F_1\underline{R_1}^{-1} & 0 & 0\\
				F_1\underline{R_1}^{-1} & 0 & 0\\
				0 & -F_2 \underline{R_2}^{-1} & 0\\
				0 & F_2 \underline{R_2}^{-1} & 0
			\end{array}\right.
		\]\vspace{2mm}\[
			\hspace{2cm}\small\left.\begin{array}{llll}
			     -\underline{R_1}^{-\top}F_1^\top & 0 & 0\\
			     0 & \underline{R_2}^{-\top}F_2^\top& -\underline{R_2}^{-\top}F_2^\top\\
			     0 & 0 & 0\\
			     0 & 0 & 0\\
			     0 & 0 & 0\\
			     0 & 0 & 0
			\end{array}\right).
		\]
	\end{lem}
	\begin{proof}
		The proof directly follows from Theorem~\ref{thm:efpesequence}, the LCP above expressing the best-response conditions of the two players in the perturbed game. However, we report the complete derivation of the LCP, being useful for our treatment.

		The problem of finding the best response of player $i$ in the perturbed game $(\Gamma,l)$ with $l(a)=\eps$ for every~$a$ is a linear problem, defined as
		\[
			\text{BR}_i(\eps)\ :\ \left\lbrace
			\begin{array}{ll}
				\max_{r_i}        & r_i^\top U_i r_{-i}\\
				\text{\normalfont{s.t.}} & F_ir_i = f_i\\
				                         & R_i(\eps)r_i \ge 0
			\end{array}\right.
		\]
		Notice that $R_i(\eps)$ is invertible (Lemma~\ref{lem:inverse R}), hence by changing variable, we find the equivalent problem:
		\[
			\text{BR}_i(\eps)\ :\ \left\lbrace
			\begin{array}{ll}
				\max_{\tilde r_i}        & \tilde r_i^\top R_i(\eps)^{-\top} U_i R_{-i}(\eps)^{-{1}}\tilde r_{-i}\\
				\text{\normalfont{s.t.}}\quad\circled{\normalfont 1} & F_i R_i(\eps)^{-1}\tilde r_i = f_i\\
				\hfill\circled{\normalfont 2}                        & \tilde r_i \ge 0
			\end{array}\right.
		\]
		Taking the dual:
		\[
			\overline{\text{BR}}_i(\eps)\ :\ \left\lbrace
			\begin{array}{ll}
				\min_{v_i}        & f_i^\top v_i\\
				\text{\normalfont{s.t.}}\quad\circled{\normalfont 3} & R_i(\eps)^{-\top} F_i^\top v_i \ge\\
				                         & \qquad R_i(\eps)^{-\top} U_i R_{-i}(\eps)^{-{1}} \tilde r_{-i}\\
				\hfill\circled{\normalfont 4}  & v_i \text{ free in sign}
			\end{array}\right.
		\]
		Complementarity slackness requires that
		\[
			\circled{\normalfont 5}\quad\tilde r_i^\top (R_i(\eps)^{-\top} F_i^\top v_i - R_i(\eps)^{-\top} U_i R_{-i}(\eps)^{-1} \tilde r_{-i}) = 0.
		\]
		Solving problem $\text{BR}_i(\eps)$ or $\overline{\text{BR}}_i(\eps)$ is equivalent to solving the feasibility problem defined by constraints \circled{1} to \circled{5}.
		It is now easy to see that we can cast the problem of satisfying conditions \circled{1} to \circled{5} for both players as a standard-form LCP whose parameters are as defined in this lemma.
	\end{proof}

	We conclude this section with a couple of lemmas that we will use in the following sections.

	\begin{restatable}{lem}{lcppolyspace}\label{lem:lcp poly bit}
		Consider the LCP formulation of Lemma~\ref{lem:EFPE LCP}, where $\eps$ is treated as a symbolic variable, so that the entries of $M(\eps)$ are polynomials in $\eps$. A number of bits polynomial in the input game size is sufficient to store all coefficients appearing in $P(\eps)$.
	\end{restatable}

	\begin{restatable}{lem}{feasibilityFRxf}\label{lem:feasibility of FRxf}
	    Let $\nu = \max_{h \in \cup_i H_i} \{|\rho(h)|\}$ be the maximum number of actions available at an information set. If $0 \le \eps \le 1/n$, there always exists a realization-plan strategy $r_i$ such that $F_i \, R_i(\eps)^{-1} \,r_i = f_i$.
	\end{restatable}

\section{Perturbed LCPs}
	Before turning our attention to the \emph{computational} aspects of finding an EFPE, we introduce some general concepts, pertaining to perturbed linear optimization problems. While we target the development of these concepts with our specific use-case in mind, it should be noted that this section's definitions and lemmas are of broader interest, being applicable to any linear program (LP) or LCP. We recall that a \emph{basis} $\mathcal{B}$ for a standard-form LP with constraints $M x = b$ or a standard-form LCP with linear equality constraints $w = Mz + b$ is a set of linearly independent columns of $M$ such that the associated solution (called \emph{basic solution}) is feasible.

	\begin{defn}[Negligible positive perturbation (NPP)]
		Let $P(\eps)$ be an LCP dependent on some perturbation $\eps$. The value $\eps^* > 0$ is a \emph{negligible positive perturbation} (NPP) if any optimal basis $\mathcal{B}$ for $P(\eps^*)$ is optimal for $P(\eps)$, for all $0 \le \eps \le \eps^*$.
	\end{defn}

	\begin{defn}[Optimality certificate for a basis]
		Given an LCP $P(\eps)$, and a basis $\mathcal{B}$ for it, we call the finite-dimensional column vector $\mathcal{C}_\mathcal{B}(\eps)$ an \emph{optimality certificate} for $\mathcal{B}$ if for all $\eps\ge 0$
		\[\mathcal{C}_\mathcal{B}(\eps) \ge 0 \iff \mathcal{B}\text{ is optimal for }P(\eps).\]
	\end{defn}

	\begin{lem}\label{lem:lcp certificate}
		In the case of a perturbed LCP in standard form 
		\[
			P(\eps)\ :\ \left\lbrace
			\begin{array}{ll}
			  \text{\normalfont{find}} & z, w\\
			  \text{\normalfont s.t.}\quad\circled{\normalfont 1}    & z^\top w = 0\\
			  \hfill\circled{\normalfont 2}    & w = M(\eps)z + b(\eps)\\
			  \hfill\circled{\normalfont 3}    & z, w \ge 0
			\end{array}\right.
		\]
		an optimality certificate for the complementary basis $\mathcal{B}$ is
		\[
			\mathcal{C}_\mathcal{B}(\eps) = B(\eps)^{-1} b(\eps)
		\]
		where $B$ is the basis matrix corresponding to $\mathcal{B}$.
	\end{lem}
	\begin{proof}
		Since the basis $\mathcal{B}$ is complementary by hypothesis, constraint \circled{1} is always satisfied. Constraint \circled{2} is satisfied by the definition of $B(\eps)$. Constraint \circled{3} is satisfied if and only if $B(\eps)^{-1} b(\eps) \ge 0$.
	\end{proof}

	Finally, before proceeding, we introduce three mathematical lemmas that come in handy when dealing with optimality certificates. Indeed, it is often the case that $\mathcal{C}_\mathcal{B}(\eps)$ has polynomial or rational functions (with respect to $\eps$) as entries.
\begin{restatable}{lem}{polynomialpositive}\label{lem:polynomial positive}
		Let $p(\eps) = a_0 + a_1 \eps^1 + \dots + a_n\eps^n$ be a real polynomial such that $a_0 \neq 0$, and let $\mu = \max_i |a_i|$. Then $p(\eps)$ has the same sign of $a_0$ for all $0 \le \eps \le \eps^*$, where $\eps^* = |a_0|/(\mu + |a_0|)$.
	\end{restatable}

	As a corollary, we can easily extend the result of Lemma~\ref{lem:polynomial positive} to rational functions.

	\begin{lem}\label{lem:rational positive}
		Let \[
			p(\eps) = \frac{a_0 + a_1 \eps^1 + \dots + a_n\eps^n}{b_0 + b_1\eps^1 + \dots + b_m\eps^m}
		\] be a rational function such that $a_0, b_0 \neq 0$, and let $\mu_a = \max_i |a_i|$, $\mu_b = \max_i |b_i|$. Then $p(\eps)$ has the same sign of $a_0/b_0$ for all $0 \le \eps \le \eps^*$, where $\eps^* = \min\{|a_0|/(\mu + |a_0|), |b_0|/(\mu+|b_0|)\}$.
	\end{lem}
	\begin{proof}
	 	The proof follows immediately by applying Lemma~\ref{lem:polynomial positive} to the numerator and the denominator of $p(\eps)$.
	\end{proof}

	\begin{restatable}{lem}{integerrationalsign}\label{lem:integer rational sign}
		Let \[
			p(\eps) = \frac{a_0 + a_1 \eps^1 + \dots + a_n\eps^n}{b_0 + b_1\eps^1 + \dots + b_m\eps^m}
		\]
		be a rational function with integer coefficients, where the denominator is not identically zero; let $\mu_a = \max_i |a_i|$, $\mu_b = \max_i |b_i|$, $\mu = \max\{\mu_a, \mu_b\}$ and $\eps^* = 1/(2\mu)$. Then exactly one of the following holds:
		\begin{itemize}
			\item $p(\eps^*) = 0$ for all $0 < \eps \le \eps^*$,
			\item $p(\eps^*) > 0$ for all $0 < \eps \le \eps^*$,
			\item $p(\eps^*) < 0$ for all $0 < \eps \le \eps^*$.
		\end{itemize}
	\end{restatable}

\section{Computation of Extensive-Form Perfect Equilibria}
	We finally delve into the computational details of finding Extensive-Form Perfect Equilibria. The central result of this section (Theorem~\ref{thm:efpe admits npp}) roughly states that the EFPE LCP (Lemma~\ref{lem:EFPE LCP}) always admits a ``small'' NPP. Leveraging this fact, we quickly derive a path-following algorithm for the computation of EFPE in general-sum games in which each pivoting step has a polynomial-time cost (Theorem~\ref{thm:general sum ppad}), and a polynomial-time algorithm for the zero-sum counterpart (Theorem~\ref{thm:zero sum fp}). These two algorithms put the two search problems in the \textsf{PPAD} and the \textsf{FP} classes, respectively.

	We start by showing that, as long as the perturbation $\eps$ is ``reasonably small'', the LCP defined in Lemma~\ref{lem:EFPE LCP} always admits a solution. In particular:
	\begin{restatable}{lem}{lcptermination}\label{lem:LCP termination}
		If $\,0 < \eps \le 1/\nu$, where $\nu = \max_{h\in\cup_i H_i} \{|\rho(h)|\}$ is the maximum number of actions available at an information set, Lemke's algorithm always finds a solution for $P(\eps)$.
	\end{restatable}
We remark that when $\eps$ is a given value, the task of finding an NE of $P(\eps)$ has a powerful interpretation. Indeed, it captures the situation in which the moves of a player are subject to execution uncertainty and therefore a player cannot perfectly control their actions.

	\begin{thm}\label{thm:efpe admits npp}
		Given a (general-sum) two-player game $\Gamma$ with $\nu = \max_{h\in\cup_i H_i} \{|\rho(h)|\}$, the problem $P(\eps)$ of determining any EFPE for $\Gamma$ admits an NPP $\eps^* \le 1/\nu$ that can be computed from $\Gamma$ in polynomial time. In particular, $\eps^* = 1/V^*$, where the integer value $V^*$ can be represented in memory with a number of bits polynomial in the input game size.
	\end{thm}
	\begin{proof}
		We illustrate the steps that lead to the determination of such $V^*$. The central idea is as follows: we want to determine $\eps^*$ so that, whatever the feasible base $\mathcal{B}$ for $P(\eps^*)$ may be, the optimality certificate for $\eps^*$ is positive for all $\eps\in(0, \eps^*]$. Indeed, it is immediate to see that such $\eps^*$ is necessarily an NPP.

		\textbf{Optimality certificate.} We begin by studying the optimality certificate for the LCP  $P(\eps)$, that is, by Lemma~\ref{lem:lcp certificate},
		\[
			B(\eps)^{-1}b(\eps) \ge 0,
		\]
		where $B(\eps)$ is base matrix corresponding to the feasible base $\mathcal{B}$ found by Lemke's algorithm. Introducing $C(\eps) = \cof B(\eps)$, the cofactor matrix of matrix $B(\eps)$, and leveraging the well-known identity $B(\eps)^{-1} = C(\eps)^\top / \det B(\eps)$, we can rewrite the optimality certificate above as
		\[
			\frac{C(\eps)^\top b(\eps)}{\det B(\eps)} \ge 0.
		\]

		The vectorial condition above is equivalent to a system of $n$ scalar conditions, each of the form
		\[
			f_i(\eps) = \frac{c_i(\eps)^{\top} b(\eps)}{\det B(\eps)} \ge 0,
		\]
		where $c_i(\eps)$ is the $i$-th row of $C(\eps)^\top$. Evidently, $f_i(\eps)$ is a rational function in $\eps$, having only integer coefficients, for all $i = 1,\dots,n$.

		\textbf{Denominator coefficients.} We now give an upper bound on the coefficients of the denominator of $f_i(\eps)$, that is $\det B(\eps)$. Let $V_B$ be the largest coefficient that could potentially appear in $B(\eps)$ and $b(\eps)$, and let $m$ be the largest polynomial degree appearing in $B(\eps)$. Notice that $m\in\bigoh(\text{poly}(n))$. By using Hadamard's inequality, we can write
		\[
			\coeff(\det B(\eps)) \le n^{n/2} V_B^n\, \coeff((1 + \eps + \dots + \eps^m)^n),
		\]
		where $\coeff(\cdot)$ is the largest coefficient of its polynomial argument. Since
		$
			\coeff((1+\eps+\dots+\eps^m)^n) \le m^n,
		$
		we have
		\[
			\coeff(\det B(\eps)) \le V_D := n^{n/2} (m V_B)^n.
		\]
		Notice that this bound is valid for all possible base matrices $B(\eps)$. Furthermore, notice that
		\[
			\log V_D = n/2 \log n + n \log m + n\log V_B,
		\]
		and by Lemma~\ref{lem:lcp poly bit} we conclude that $V_D$ requires a number of bits polynomial in the input game size in order to be stored in memory.

		\textbf{Numerator coefficients.} Since the elements of $c_i(\eps)$ are cofactors for $B(\eps)$, they are upper-bounded by $\det B(\eps)$, which in turn is upper-bounded by $V_D$. Therefore,
		\[
			\coeff(c_i(\eps)^\top b(\eps)) \le V_N := V_B V_D.
		\]
		Again, it is worthwhile to notice that this bound is valid for all possible base matrices $B(\eps)$.

		\textbf{Wrapping up.} Define $V^* = 2\max\{V_N, V_D\} = 2 V_B V_D$. We now argue that $\eps^* = 1/V^*$ is an NPP for $P(\eps)$. Indeed, let $\mathcal{B}^*$ be a feasible base\footnote{Notice that since $\eps^* \le 1/n$, $\mathcal{B}$ always exists (see Lemma~\ref{lem:LCP termination}).} for $P(\eps^*)$, and let $B(\eps^*)$ be the corresponding base matrix. Being $\mathcal{B}$ feasible for $P(\eps^*)$, each row $f_i$ in the optimality certificate is non-negative when evaluated at $\eps^*$ for all $i$. Therefore, we know from Lemma~\ref{lem:integer rational sign} that $f_i(\eps) \ge 0$ in $(0, 1/V^*] = (0, \eps^*]$. Hence, the optimality certificate for $\mathcal{B}$ is non-negative for all $0 < \eps \le \eps^*$, which is equivalent to say that $\eps^*$ is an NPP. Finally, note that $V^* = 2V_B V_D$ be stored in memory with a number of bits polynomial in the game size. This completes the proof.
	\end{proof}
	
	\begin{thm}\label{thm:general sum ppad}
	    The problem of determining an EFPE of a general-sum two-player game $\Gamma$ is \textsf{\emph{PPAD}}-complete.
	\end{thm}
	\begin{proof}
		Let $\eps^* = 1/V^*$ be an NPP as defined in Theorem~\ref{thm:efpe admits npp}, and let $\mathcal{B}$ be a feasible base for the (numerical) problem $P(\eps^*)$, found using Lemke's algorithm. Since $\eps^*$ is an NPP, the pair of strategies $(\pi_1^*, \pi_2^*)$ corresponding to $\mathcal{B}$ retain their feasibility with respect to the LCP $P(\eps)$ as $\eps \to 0$, meaning that $(\pi_1^*, \pi_2^*)$ is in fact an EFPE.

		Furthermore, given that $V^*$ requires a number of bits polynomial in the input game size, each iteration of Lemke's algorithm takes time polynomial in the game size. This proves that the algorithm described is  a path-following algorithm requiring a polynomial-time cost at each step, and therefore the problem of finding an EFPE in two-player games is in the \textsf{PPAD} class. The hardness easily follows from the fact that 
 EFPE is a refinement of Nash equilibrium, 
 an EFPE always exists, and 
 finding a Nash is \textsf{PPAD}-complete.
 Therefore, if finding an EFPE were not \textsf{PPAD}-hard, then one could use the EFPE-finding algorithm with the aim of finding an NE and therefore not even finding an NE would be \textsf{PPAD}-hard. This concludes the proof.
	\end{proof}
We remark that the proof of theorem above also applies for an arbitrary $\eps$ (potentially non-NPP), showing that finding an NE for any $\eps< \max_{h\in\cup_i H_i} \{|\rho(h)|\}$ is in the \textsf{PPAD} class. We summarize the procedure to find an EFPE of general-sum games in Algorithm~\ref{algo:main}.

	\begin{algorithm}
		\begin{algorithmic}
			\Procedure{Find-EFPE}{}
				\State 1. Compute $\eps^*$ from $\Gamma$ as in the proof of Theorem~\ref{thm:efpe admits npp}
				\State 2. Determine a basis $\mathcal{B}$ for the numerical LCP $P(\eps^*)$
				\State 3. Let $B(\eps)$ be the base matrix corresponding to $\mathcal{B}$ in $P(\eps)$, as $\eps$ varies.
				\State \textcolor{gray}{$\triangleright$ Since $B(\eps)^{-1} b$ is a rational bounded function in a neighborhood of $0$, $B(0)^{-1} b$ exists.}
				\State 4. $(\tilde r_1, \tilde r_2, v_1^+, v_1^-, v_2^+, v_2^-)^\top = B(0)^{-1} b$
				\State {\textcolor{gray}{$\triangleright$ Note that $R_1(0) = R_2(0) = I$, so $\tilde r_1 = r_1, \tilde r_2 = r_2$}}
				\State 5. \Return the pair of strategies $(\tilde r_1, \tilde r_2)$ 
			\EndProcedure
		\end{algorithmic}
		\caption{}
		\label{algo:main}	
	\end{algorithm}

The approach we use in Theorems~\ref{thm:efpe admits npp} and~\ref{thm:general sum ppad} extends the one used in~\cite{quasiperfect}. More precisely, in~\cite{quasiperfect} the authors consider a numerical perturbation that sums to the constant terms of an LP to find a QPE of a zero-sum game, while in Theorem~\ref{thm:efpe admits npp} we consider perturbations over the coefficient of the variables of an LCP to find an EFPE of general-sum games (and, below, of zero-sum games). The two approaches can be extended to find a QPE in general-sum games by using a numerical perturbation (the description is omitted here, as it is beyond the scope of this paper). We now show that Algorithm~\ref{algo:main} requires polynomial time when the game is zero sum.
	\begin{thm}\label{thm:zero sum fp}
		The problem of determining an EFPE of a zero-sum two-player game $\Gamma$ can be solved in polynomial time in the size of the input game.
	\end{thm}
	\begin{proof}
		Like in Theorem~\ref{thm:general sum ppad}, we can easily extract an EFPE by looking at the feasible matrix $\mathcal{B}$ which solves the (numerical) LCP $P(\eps^*)$. However, in the zero-sum setting, we do not need to use Lemke's algorithm. Indeed, notice that in zero-sum two-player games, matrix $M(\eps)$ as defined in Lemma~\ref{lem:EFPE LCP} is such that $M(\eps) + M(\eps)^\top = 0$ for all $\eps$, because $U_2 = -U_1^\top$. Therefore, the complementarity condition can be rewritten as
		\[
			\begin{aligned}
				z^\top(M(\eps)z + b(\eps)) &= z^\top M(\eps)z + z^\top b(\eps)\\
				                           &= \frac{1}{2}\,z^\top\left(M(\eps) + M(\eps)^\top\right) z + z^\top b(\eps)\\
				                           &= z^\top b(\eps) = 0,
			\end{aligned}
		\]
		a linear condition instead of a quadratic one. This shows that when the game is zero-sum, the LCP is actually an LP. As such, a basis for the LCP of Algorithm~\ref{algo:main} can be computed in polynomial time, leading to an overall polynomial time algorithm. This completes the proof.
	\end{proof}

\section{Conclusion and Future Work}
	In this paper, we provide a path-following algorithm to find an EFPE in 2-player games. Our algorithm requires the application of Lemke's algorithm to a numerically perturbed LCP. We show that the computation cost of each iteration of the algorithm is polynomial, and this shows that finding an EFPE in 2-player games is \textsf{PPAD}-complete. We also show that in the notable case of 2-player zero-sum games, linear programming can be used and that the problem is in the \textsf{FP} class. In order to achieve our result, we also develop two accessory results. The first one shows that the problem of finding a Nash equilibrium when a player does not perfectly control her moves being subject to mistakes, as it happens in practice for physical agents, is \textsf{PPAD}-complete and can be done by means of our algorithm. The second one is an extension of the characterization of numerically perturbed LCPs in which even the coefficients of the variables are perturbed.

	In future works, we aim to extend our accessory results as well as to study the verification problem for an EFPE in 2-player games (that is, the problem of deciding whether a strategy profile given in input is an EFPE).

\bibliography{citations}
\bibliographystyle{aaai}
\clearpage

\appendix
\section{Appendix}
\vspace{1cm}
\section{Discussion of (Gatti and Iuliano 2011)}

We provide a brief discussion about the previous results about the computation of an EFPE in 2-player games.
\begin{rem}\label{oldpaper}
In~\cite{DBLP:conf/aaai/GattiI11}, the authors provide two versions of Lemke's algorithm applied to the sequence form, claiming that they compute an EFPE and that the computational cost of each iteration is polynomial. We initially observe that the authors do not provide any proof of the soundness of the their algorithms and  of the polynomial cost of each single step of the algorithm.

For the sake of presentation, we initially provide a discussion about the second version of the algorithm (described in the section titled ``Finding an EFPE in Non-Uniform $\varepsilon$-Perturbed Games''). Here the authors propose the adoption of a perturbation in the dual best-response constraints defined as follows: for simplicity we report only the perturbation for player~1, it is $-\eps^{\frac{|q_{\max}|+1}{|q|}}$ where $q_{\max}$ is the longest sequence of player~$1$. Basically, every time the same utility is reached at different terminal nodes, they prefer the terminal node reached with the smallest sequence. We can show that such a perturbation may not lead to any EFPE. Consider the game in Figure~\ref{fig:counterexample}. There is only one player, say player~1. Any EFPE of this game prescribes player~1 to play~$\mathsf{R}_1$. Indeed, at information set~$1.1$, the expected utility of player~1 from playing~$\mathsf{R}_1$ is 1, while the expected utility from playing~$\mathsf{L}_1$ is strictly smaller than 1 (the exact value depends on the perturbation used). By using the perturbation above, the algorithm returns $\mathsf{L}_1\mathsf{L}_2$, since the value of such sequence is $1-\eps^2$, while the value of any terminal sequence of the form `$\mathsf{R}_1* *$' is $1-\eps$. Since $\eps$ goes to zero, $1-\eps^2>1-\eps$.
	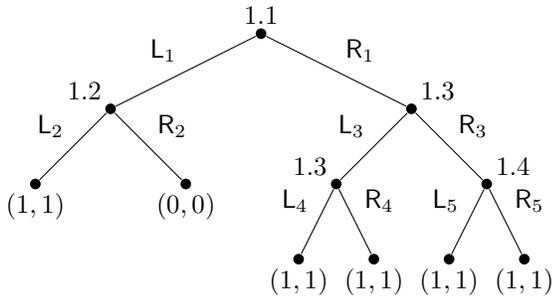
\begin{figure}[H]
		\centering\begin{tikzpicture}[auto]
			\node[circle, inner sep=.5mm, fill=black] (root) at (0, 0) {};
			\node[circle, inner sep=.5mm, fill=black] (l1) at (-2, -1) {};
			\node[circle, inner sep=.5mm, fill=black] (r1) at (2, -1) {};
			\node[circle, inner sep=.5mm, fill=black] (l1l2) at (-3, -2) {};
			\node[circle, inner sep=.5mm, fill=black] (l1r2) at (-1, -2) {};
			\node[circle, inner sep=.5mm, fill=black] (r1l3) at (1, -2) {};
			\node[circle, inner sep=.5mm, fill=black] (r1r3) at (3, -2) {};
			\node[circle, inner sep=.5mm, fill=black] (r1l3l4) at (0.5, -3) {};
			\node[circle, inner sep=.5mm, fill=black] (r1l3r4) at (1.5, -3) {};
			\node[circle, inner sep=.5mm, fill=black] (r1r3l5) at (3.5, -3) {};
			\node[circle, inner sep=.5mm, fill=black] (r1r3r5) at (2.5, -3) {};
			\draw (root) node[above]{$1.1$} --node[above left] {$\mathsf{L}_1$} (l1) node[above left] {$1.2$};
			\draw (root) --node[above right] {$\mathsf{R}_1$} (r1) node[above right] {$1.3$};
			\draw (l1) --node[above left] {$\mathsf{L}_2$} (l1l2) node[below] {$(1,1)$};
			\draw (l1) --node[above right] {$\mathsf{R}_2$} (l1r2) node[below] {$(0,0)$};
			\draw (r1) --node[above left] {$\mathsf{L}_3$} (r1l3) node[above left] {$1.3$};
			\draw (r1) --node[above right] {$\mathsf{R}_3$} (r1r3) node[above right] {$1.4$};
			\draw (r1l3) --node[above left] {$\mathsf{L}_4$} (r1l3l4) node[below] {$(1,1)$};
			\draw (r1l3) --node[above right] {$\mathsf{R}_4$} (r1l3r4) node[below] {$(1,1)$};
			\draw (r1r3) --node[above right] {$\mathsf{R}_5$} (r1r3l5) node[below] {$(1,1)$};
			\draw (r1r3) --node[above left] {$\mathsf{L}_5$} (r1r3r5) node[below] {$(1,1)$};
		\end{tikzpicture}
		\caption{A game used as counterexample in Remark~\ref{oldpaper}.}
		\label{fig:counterexample}
	\end{figure}

The analysis of the first version of the algorithm (described in the section titled ``Finding an EFPE in Uniform $\varepsilon$-Perturbed Games'') is more involved and we provide just a sketch. First, the authors propose a double perturbation---an additive one as proposed by~\cite{quasiperfect} and a new one that is multiplicative---and they claim that adopting these perturbations is equivalent to consider a perturbed game $(\Gamma,l)$ where $l(a)=\varepsilon$ for every~$a$. However, this is not true as shown in our paper where we show the perturbation over the sequences leading to such a $(\Gamma, l)$. The perturbed LCP we provide in our paper and that one provided in~\cite{DBLP:conf/aaai/GattiI11} are different, e.g., in our LCP even the sequence-form constrains $F_i r_i = f_i$ are subject to a multiplicative perturbation, while in the LCP provided in~\cite{DBLP:conf/aaai/GattiI11} those constraints do not present any multiplicative perturbation. Nevertheless, we tried to look for a simple counterexample showing that the perturbations proposed in~\cite{DBLP:conf/aaai/GattiI11} fail in finding an EFPE, but we did not find it. Second, in the algorithm proposed by the authors, each coefficient of matrix $M$ of the LCP is subject to a symbolic perturbation expressed as a polynomial in~$\varepsilon$ whose maximum degree increases at each iteration. The crucial issue is that the increase is exponential, and therefore the maximum degree of the polynomial increases exponentially, requiring to store an exponential amount of numbers. The authors use the integer pivoting in their algorithm. When integer pivoting is used, e.g., in the simplex algorithm, the values of the numbers stored in the tableau rise exponentially, but they can be stored with a linear number of bits by using binary representation. Conversely, in our case, since the maximum degree of the polynomial rises exponentially, we need to store an exponentially large number of coefficients. We cannot exclude the case in which some coefficients can be discarded keeping only a polynomial number of coefficients, but no proof is provided in~\cite{DBLP:conf/aaai/GattiI11} and we did not find any  simple way to prove that.
\end{rem}

\section{Lemke's algorithm conditions}

	Lemke's algorithm~\cite{lemke1970recent} is an iterative algorithm able to solve a linear complementarity problem, provided it satisfies the following conditions:

	\begin{lem}[Theorem 4.1, Koller, Megiddo and von Stengel 1996]\label{lem:LCP conditions}
		If:\\[1mm]
		\begin{tabular}{ll}
			(a) & $z^\top M z \ge 0$ for all $z \ge 0$, and\\
			(b) & $z\ge 0, Mz \ge 0, z^\top M z = 0 \implies z^\top b \ge 0$,
		\end{tabular}\\[1mm]
		then Lemke's algorithm computes a solution of the LCP and does not terminate with a secondary ray.
	\end{lem}

\section{Omitted proofs}

	\inverseR*
	\begin{proof}
		By induction on $n$. The lemma trivially holds for $n=1$. Now, suppose the theorem holds for $n = \bar n$; we will show that it holds for the $(\bar n + 1)\times(\bar n + 1)$ behavioral matrix $R(\eps)$. Indeed, we have
		\[
		  R(\eps) = \left(\begin{array}{c|l}
            R'(\eps) & 0\\
		    \hline\\[-3mm]
		    b(\eps)^\top & 1
		  \end{array}\right),\,	b(\eps)^\top = (0,\dots, 0, {-\eps}, 0, \dots, 0).
		\]
		where $R'(\eps)$ is a $\bar n \times \bar n$ behavioral perturbation matrix. Hence, the matrix
		\[
		R(\eps)^{-1} = \left(\begin{array}{c|l}
		            R'(\eps)^{-1} & 0\\
				    \hline\\[-3mm]
				    -b(\eps)^\top R'(\eps)^{-1}  & 1
				  \end{array}\right)
		\]
		is indeed the inverse matrix of $R(\eps)$. Using the inductive hypothesis, we have $R'(\eps)^{-1} = I' + \eps E'(\eps)$, and therefore
		\[
		  R(\eps)^{-1} = I + \eps\left(\begin{array}{c|l}
		  		            E'(\eps) & 0\\
		  				    \hline\\[-3mm]
		  				    (-b(\eps)/\eps)^\top R'(\eps)^{-1}  & 0
		  				  \end{array}\right).
		\]
		Finally, note that $(-b(\eps)/\eps)^\top = (0,\dots,0,1,0\dots,0)$ is a non-negative real vector. Therefore, $E''(\eps) = (-b(\eps)/\eps)^\top R'(\eps)^{-1}$ is a row vector whose entries are polynomials with non-negative integer coefficients, so that
		\[
		  R(\eps)^{-1} = I + \eps\left(\begin{array}{c|l}
		  		  		            E'(\eps) & 0\\
		  		  				    \hline\\[-3mm]
		  		  				    E''(\eps)  & 0
		  		  				  \end{array}\right) = I + \eps E(\eps),
		\]
		where $E(\eps)$ is a lower triangular matrix whose entries are polynomials in $\eps$ having non-negative integer coefficients, as we wanted to prove.
	\end{proof}

	\lcppolyspace*
	\begin{proof}
		Consider the LCP formulation of Lemma~\ref{lem:EFPE LCP}. We begin by showing that each coefficient appearing in $P(\eps)$ requires a polynomial amount of memory to be store. This property trivially holds for vector $b$. On the other hand, all numbers appearing in matrix $M(\eps)$ are either zeros, or they are obtained by multiplying two or more of the following matrices together: $R_1^{-\top}, R_2^{-\top}, U_1, U_2^\top, F_1^\top, F_2^\top$. Hence, as long as each of the coefficients appearing in the above-mentioned matrices requires a polynomial number of bits in the input game size, the property is true. This is clearly true for $U_1, U_2^\top, F_1, F_2^\top$, so we are left with the task of proving this property for $R_1(\eps)^{-1}$ and $R_2(\eps)^{-1}$. However, since $\det R_1(\eps) = 1$ (indeed, notice that $R_1(\eps)$ is lower triangular), using the adjoint matrix theorem and the Leibniz formula for the determinant, we conclude that each entry in $R_1(\eps)^{-1}$ is obtained as a sum of $n!$ terms, each of which is a product of $n$ entries of $R_1(\eps)$, where $n$ is the size of $R_1(\eps)$ (the same holds for $R_2(\eps)$). Therefore, the property holds, showing that each coefficient in $M(\eps)$ and $b$ requires a polynomial amount of memory to be stored.

		We now show that the maximum degree appearing in $M(\eps)$ is $2n$. This is a consequence of the observation above: since each entry in $R_1(\eps)^{-1}$ is obtained as sum of $n!$ terms, each of which is a product of $n$ entries of $R_1(\eps)$, the maximum degree appearing in $R_1(\eps)^{-1}$ is $n$, where $n$ is the size of $R_1(\eps)$ (the same holds for $R_2(\eps)$). Now, since each element of $M(\eps)$ is obtained from the product of at most two matrices dependent on $\eps$, the maximum degree appearing in $M(\eps)$ (and therefore in $P(\eps)$) is $2n$.

		Thus, we have a polynomial amount of coefficients to store, each of which requires a polynomial amount of memory. The required space is therefore polynomial.
	\end{proof}

	\feasibilityFRxf*
	\begin{proof}
	 	We will prove that there always exists a realization-plan strategy profile $y$ such that $R_i(\eps) y \ge 0$ when $0 \le \eps \le 1/\nu$. This statement is equivalent to that of the lemma.

	 	We let such realization-plan strategy profile $y$ be defined as follows:
	 	\[
	 	  y(q_\varnothing) = 1,\quad y(qa) = \frac{y(q)}{|\rho(h)|},
	 	\]
	 	where $h$ is the information set to which $q$ leads. It is immediate to see that such $y$ is indeed a realization-plan strategy profile, that is $F_iy = f_i$. Indeed,
	 	\[
	 	  \sum_{a\in\rho(h)} y(qa) = \sum_{a\in\rho(h)} \frac{y(q)}{|\rho(h)|} = y(q).
	 	\]

	 	We now prove that $y$ is such that $R_i(\eps) y \ge 0$ for all $0 < \eps \le 1/\nu$. Indeed, notice that because of the peculiar structure of $R_i(\eps)$, the condition $R(\eps)y \ge 0$ is actually equivalent to
	 	\[
	 	  y(qa) \ge \eps y(q),\quad \forall q.
	 	\]
	 	Since $|\rho(h)| \le \nu$ for all $q$ and $0 < \eps \le 1/\nu$ by hypothesis, we have $1/|\rho(h)| \ge \eps$ for all $q$, and the inequality above holds. This completes the proof.
	\end{proof}

	\polynomialpositive*
		\begin{proof}
			We prove that when $a_0 > 0$, $p(\eps)$ is positive for all $0 \le \eps \le \eps^*$.
			Indeed,
			\[
				\begin{array}{rl}
					p(\eps) &= \displaystyle a_0 + a_{1}\eps^{1} + \dots + a_n\eps^n \\
					        &> \displaystyle a_0 - \mu \eps \sum_{i = 0}^{\infty} \eps^{i}\\
					        &= \displaystyle a_0 - \frac{\mu\eps}{1-\eps}.
				\end{array}
			\]
			Since $\eps \le \eps^* = a_0/(\mu + a_0)$ we have
			\[
				p(\eps) > a_0 - \frac{\mu a_0}{\mu+a_0-a_0} = 0.
			\]

			To conclude the proof, we need to show that when $a_0 < 0$, $p(\eps)$ is negative for all $0 \le \eps \le \eps^*$. The proof of this part is completely symmetric to that of the previous part.
		\end{proof}

	\integerrationalsign*
		\begin{proof}
			If the numerator of $p(\eps)$ is identically zero, the thesis follows trivially, as $p(\eps) = 0$ for all $\eps$, while the denominator is never zero for all $0<\eps\le\eps^*$ due to Lemma~\ref{lem:polynomial positive}. If, on the other hand, the numerator of $p(\eps)$ is not identically zero, there exist $q$ and $r$, both non-negative, such that $p$ can be written as
			\[
			p(\eps) = \frac{\eps^q(a_{q} + a_{q+1}\eps + \dots + a_{n}\eps^{n-q})}{\eps^r(b_r + b_{r+1}\eps^1 + \dots + b_{n}\eps^{n-r})},
			\]
			with $a_q, b_r \neq 0$. Since $a_q$ and $b_r$ are integer, $|a_q|, |b_r| \ge 1$ and we have
			\[
			\eps^* = \frac{1}{2\mu} \le \min\left\{\frac{|a_q|}{\mu_a + |a_q|}, \frac{|b_r|}{\mu_b + |b_r|}\right\}.
			\]
			Using Lemma~\ref{lem:rational positive} we conclude that the sign of $p(\eps)$ is constant, and equal to that of $a_q/b_r$, for all $0 < \eps \le \eps^*$.
		\end{proof}

	\lcptermination*
	\begin{proof}
		We follow the same proof structure as that in \cite[Section 4]{koller1996efficient}. In particular, we prove that if $U_1, U_2 < 0$, then conditions (a) and (b) of Lemma~\ref{lem:LCP conditions} hold for all problems $P(\eps)$ defined in Lemma~\ref{lem:EFPE LCP}. Notice that we can always assume $U_1, U_2 < 0$ without loss of generality, as we can apply an offset to the payoff matrices leaving the game unaltered.\\[1mm]

		\textbf{Condition (a)}. We need to show that when $U_1, U_2 < 0$, then $z^\top M(\eps) z \ge 0$ for all $z\ge 0$. We have:
		\[
			z^\top M(\eps) z = {\tilde r_1}^\top R_1(\eps)^{-\top} (-U_1 - U_2) R_2(\eps)^{-1} {\tilde r_2}.
		\]
		Substituting $U = -U_1 - U_2 > 0$ and using Lemma~\ref{lem:inverse R}:
		\begin{equation}\label{eq:condition A}
			\begin{array}{rl}
				z^\top M(\eps) z &= {\tilde r_1}^\top (I + \eps E_1(\eps)^\top) U (I + \eps E_2(\eps)){\tilde r_2}\\
				                 &\ge {\tilde r_1}^\top U {\tilde r_2}.
			\end{array}
		\end{equation}
		When $z \ge 0$, then $\tilde r_1, \tilde r_2 \ge 0$ and we conclude that $\tilde r_1^\top U \tilde r_2 \ge 0$, which implies the thesis.\\[1mm]

		\textbf{Condition (b)}. We already proved (Equation~\ref{eq:condition A}) that
		\[
			z^\top M(\eps) z \ge \tilde r_1^\top U \tilde r_2,
		\]
		where $U > 0$. In order for $z^\top M(\eps) z$ to be zero given $z \ge 0$, it is  necessary that $\tilde r_1, \tilde r_2 = 0$. Defining $v_1 = v_1^+ - v_1^-$ and $v_2 = v_2^+ - v_2^-$, we have
		\[
			M(\eps)z = \begin{pmatrix}
				R_1(\eps)^{-\top} F_1^{\top} v_1\\
				R_2(\eps)^{-\top} F_2^{\top} v_2\\
				0\\ 0 \\ 0 \\ 0
			\end{pmatrix},\ \
			z^\top b = b^\top z = f_1^\top v_1 + f_2^\top v_2.
		\]
		Hence, in order to complete the proof, it suffices to show that
		\[
			R_i(\eps)^{-\top} F_i^\top v_i \ge 0 \implies f_i^\top v_i \ge 0\qquad(i\in\{1,2\}).
		\]
		To this end, we consider the following linear optimization problem $Y_i(\eps)$, and its dual $\bar Y_i(\eps)$:
		\[
			\begin{array}{rcl}
			Y_i(\eps) &:&\left\lbrace\begin{array}{ll}
				\max_{\tilde r_i}        & 0\\
				\text{\normalfont{s.t.}} & F_i R_i(\eps)^{-1}\tilde r_i = f_i\\
				                         & \tilde r_i \ge 0
			\end{array}\right.,\\[7mm]
			\bar Y_i(\eps)&:&\left\lbrace\begin{array}{ll}
							\min_{v_i}        & f_i^\top v_i\\
							\text{\normalfont{s.t.}} & R_i(\eps)^{-\top}F_i^\top v_i \ge 0\\
						\end{array}\right..
			\end{array}
		\]
		Notice that $Y_i(\eps)$ is feasible since $\eps \le 1/\nu$ by hypothesis (Lemma~\ref{lem:feasibility of FRxf}). Indeed, with such an $\eps$ the induced perturbed game $(\Gamma,l)$ is such that $\sum_{a \in \rho(h)}l(a)<1$ for every~$h$. By the strong duality theorem, we conclude that whenever the constraint of the dual problem is satisfied, the objective value is non-negative, that is $R_i(\eps)^{-\top} F_i^\top v_i \ge 0 \implies f_i^\top v_i \ge 0$ as we aimed to show.
	\end{proof}
\end{document}